\documentclass[english,11pt]{article}
\usepackage[english]{babel}
\usepackage{latexsym,amssymb,amsmath,amsfonts,amscd,epsfig,amsthm,color}
\usepackage{graphicx}
\usepackage{float}
\usepackage[left=1in,right=1in]{geometry}
\usepackage{framed}
\usepackage[leftbars,color]{changebar}
\usepackage{ifthen}
\usepackage{authblk}
\usepackage{colonequals}

\newboolean{ElectronicVersion}
\setboolean{ElectronicVersion}{true}

\ifthenelse{\boolean{ElectronicVersion}}{
    \usepackage[letterpaper=true,pdftex,bookmarks,pagebackref,
	plainpages=false, 
        pdfpagelabels=true 
        ]{hyperref}}{}

\makeatletter
\newtheorem*{rep@theorem}{\rep@title}
\newcommand{\newreptheorem}[2]{%
\newenvironment{rep#1}[1]{%
 \def\rep@title{#2 \ref*{##1}}%
 \begin{rep@theorem}}%
 {\end{rep@theorem}}}
\makeatother

\newtheorem{theorem}{Theorem}[section]

\newtheorem{corollary}[theorem]{Corollary}

\newreptheorem{theorem}{Theorem}
\newreptheorem{lemma}{Lemma}

\newcommand{\hlight}[1]{\ \\ \fcolorbox{black}{white}{\begin{minipage}{6.5in}{\color{black} #1}\end{minipage}} \\}

\renewenvironment{proof}{\setlength{\parindent}{\parindent}\MakeFramed {\noindent\emph{Proof.}}{}}
{\qed\endMakeFramed} 

\def\ket#1{{\lvert}#1\rangle}

\def\abs#1{\left| #1 \right|}

\def\cost#1{\mathfrak{#1}}
\def\a{f_A}
\def\b{f_B}
\newcommand{\eps}{\varepsilon}
\def\O{\mathrm{O}}
\def\tO{\widetilde{\mathrm{O}}}
\def\set#1{\mathcal{#1}}
\renewcommand{\th}[1]{${#1}^{\textrm{th}}$}
\renewcommand{\(}{\left(}
\renewcommand{\)}{\right)}
\newcommand{\defeq}{\colonequals}
\newcommand{\BSR}{\mathbb{B}}
\def\m{m}

\title{Improving Quantum Query Complexity of Boolean Matrix Multiplication Using Graph Collision}
\author[1,2]{Stacey Jeffery\thanks{\texttt{sjeffery@uwaterloo.ca}}}
\author[1,2]{Robin Kothari\thanks{\texttt{rkothari@cs.uwaterloo.ca}}}
\author[3]{Fr\'ed\'eric Magniez\thanks{\texttt{frederic.magniez@univ-paris-diderot.fr}}}
\affil[1]{\small David R.\ Cheriton School of Computer Science, University of Waterloo, Canada}
\affil[2]{Institute for Quantum Computing, University of Waterloo, Canada}
\affil[3]{\small LIAFA, Univ.\ Paris Diderot, CNRS; Paris, France}

\date{}

\begin{document}

\maketitle

\begin{abstract}
The quantum query complexity of Boolean matrix multiplication is typically studied as a function of the matrix dimension, $n$, as well as the number of $1$s in the output, $\ell$. We prove an upper bound of $\tO(n\sqrt{\ell})$ for all values of $\ell$. This is an improvement over previous algorithms for all values of $\ell$. On the other hand, we show that for any $\eps < 1$ and any $\ell \leq \eps n^2$, there is an $\Omega(n\sqrt{\ell})$ lower bound for this problem, showing that our algorithm is essentially tight.

We first reduce Boolean matrix multiplication to several instances of graph collision. We then provide an algorithm 
that takes advantage of the fact that the underlying graph in all of our instances is very dense to find all graph collisions efficiently. 
\end{abstract}

\section{Introduction}
Quantum query complexity has been of fundamental interest since the inception of the field of quantum algorithms \cite{BBBV97,gro96,sho97}.
The quantum query complexity of Boolean matrix multiplication was first studied by Buhrman and \v{S}palek~\cite{BS06}. In the Boolean matrix multiplication problem, we want to multiply two $n\times n$ matrices $A$ and $B$ over the Boolean semiring, which consists of the set $\{0,1\}$ with  logical \textsc{or} ($\vee$) as the addition operation and logical \textsc{and} ($\wedge$) as the multiplication operation.

For this problem it is standard to consider an additional parameter in the complexity: The number of $1$s in the product $C\defeq AB$, which we denote by $\ell$. We study the query complexity as a function of both $n$ and $\ell$, and obtain improvements for all values of $\ell$. 

The problem of Boolean matrix multiplication is of fundamental interest, in part due to its relationship to a variety of graph problems, such as the triangle finding problem and the all-pairs shortest path problem. 

Classically, it was shown by Vassilevska Williams and Williams that ``practical advances in triangle detection would imply practical [Boolean matrix multiplication] algorithms'' \cite{VW10}. The previous best quantum algorithm for Boolean matrix multiplication, by Le Gall, is based on a subroutine for finding triangles in graphs \emph{with a known tripartition} \cite{leg12}, already suggesting that the relationship between Boolean matrix multiplication and triangle finding might be more complex for quantum query complexity. We give further evidence for this by bypassing the triangle finding subroutine entirely. 

Despite its fundamental importance, much has remained unknown about the quantum query complexity of Boolean matrix multiplication and its relationship with other query problems in the quantum regime. Even for the simpler decision problem of Boolean matrix product verification, where we are given oracle access to three $n\times n$ Boolean matrices,  $A$, $B$ and $C$, and must decide whether or not $AB=C$, the quantum query complexity is unknown. The best upper bound is $O(n^{3/2})$ \cite{BS06}, whereas the lower bound was recently improved from the trivial $\Omega(n)$ to $\Omega(n^{1.055})$ by Childs, Kimmel, and Kothari~\cite{CKK11}. 

A better understanding of these problems may lead to an improved understanding of quantum query complexity in general. We contribute to this by closing the gap (up to logarithmic factors) between the best known upper and lower bounds for Boolean matrix multiplication for all $\ell\leq \eps n^2$ for any constant $\eps<1$. 

\paragraph{Previous Work.} We are interested in the query complexity of Boolean matrix multiplication, where we count the number of accesses (or \emph{queries}) to the input matrices $A$ and $B$. Buhrman and \v{S}palek~\cite[Section 6.2]{BS06} describe how to perform Boolean matrix multiplication using $\tO(n^{3/2}\sqrt{\ell})$ queries, by simply quantum searching for a pair $(i,j)\in[n]\times[n]$ such that there is some $k\in[n]$ for which $A[i,k]=B[k,j]=1$, where $[n]=\{1,\dots,n\}$. By means of a classical reduction relating Boolean matrix multiplication and triangle finding, Vassilevska Williams and Williams~\cite{VW10} were able to combine the quantum triangle finding algorithm of Magniez, Santha and Szegedy~\cite{MSS07} with a classical strategy of Lingas~\cite{lin09} to get a quantum algorithm for Boolean matrix multiplication with query complexity $\tO(\min\{n^{1.3}\ell^{17/30},n^2+n^{13/15}\ell^{47/60}\})$. 

Recently, Le Gall \cite{leg12} improved on their work by noticing that the triangle finding needed for Boolean matrix multiplication involves a tripartite graph with a known tripartition. He then recast the known quantum triangle finding algorithm of \cite{MSS07} for this special case and improved the query complexity of Boolean matrix multiplication. He then further improved the algorithm for large $\ell$ by adapting the strategy of Lingas to the quantum setting. 

\paragraph{Our Contributions.} Since previous quantum algorithms for Boolean matrix multiplication are based on a triangle finding subroutine, a natural question to ask is whether triangle finding is a bottleneck for this problem. We show that this is not the case by bypassing the triangle finding problem completely to obtain a nearly tight result for Boolean matrix multiplication.

A key  ingredient of the best known quantum algorithm for triangle finding is an efficient algorithm for the graph collision problem. Our main contribution is to build an algorithm directly on graph collision instead, bypassing the use of a triangle finding algorithm. Surprisingly, we do not use the graph collision algorithm that is used as a subroutine in the best known quantum algorithm for triangle finding. That algorithm is based on Ambainis' quantum walk for the element distinctness problem~\cite{amb04}. Our algorithm, on the other hand, does not have any quantum walks.

There are two main ideas we would like to impress upon the reader. First, we can reduce the Boolean matrix multiplication problem to several instances of the graph collision problem. Second, the instances of graph collision that arise depend on $\ell$; in particular, they have at most $\ell$ non-edges. Furthermore, we need to find all graph collisions, not just one. We provide an algorithm to find a graph collision in query complexity $\tO(\sqrt{\ell}+\sqrt{n})$, or to find all graph collisions in time $\tO(\sqrt{\ell}+\sqrt{n\lambda})$, where $\lambda\geq 1$ is the number of graph collisions. Combining these ideas yields the aforementioned $\tO(n\sqrt{\ell})$ upper bound.

A lower bound of $\Omega(n\sqrt{\ell})$ for all values of $\ell\leq\eps n^2$ for any constant $\eps<1$ follows from a simple reduction to $\ell$-\textsc{Threshold}, which we state formally in Theorem \ref{th:lb}. 

This paper is organized as follows. After presenting some preliminaries in Section \ref{sec:prelim}, we describe in Section \ref{sec:graph-col}, the graph collision problem, its relationship to Boolean matrix multiplication, and a subroutine for finding all graph collisions when there are at most $\ell$ non-edges. In Section \ref{sec:bmm}, we apply our graph collision subroutine to get the stated upper bound for Boolean matrix multiplication, and then describe a tight lower bound that applies to all values of $\ell\leq \eps n^2$ for $\eps<1$.  

\section{Preliminaries}\label{sec:prelim}
\subsection{Quantum Query Framework}

For a more thorough introduction to the quantum query model, see \cite{BBCMdW01}. For Boolean matrix multiplication, we assume access to two query operators that act as follows on a Hilbert space spanned by $\{\ket{i,j,b}:i,j\in[n],b\in\{0,1\}\}$:
$$\mathcal{O}_A:\ket{i,j,b}\mapsto\ket{i,j,b\oplus A[i,j]}\;\;\mathcal{O}_B:\ket{i,j,b}\mapsto\ket{i,j,b\oplus B[i,j]}$$
In the quantum query model, we count the uses of $\mathcal{O}_A$ and $\mathcal{O}_B$, and ignore the cost of implementing other unitaries which are independent of $A$ and $B$. We call $\mathcal{O}_A$ and $\mathcal{O}_B$ the \emph{oracles}, and each access a \emph{query}. The query complexity of an algorithm is the maximum number of oracle accesses used by the algorithm, taken over all inputs.

We denote a problem \textsc{P} by a map $\mathcal{X}\to2^\mathcal{Y}$, where $\mbox{\textsc{P}}(x)\subseteq \mathcal{Y}$ denotes the set of valid outputs on input $x$.
We say a quantum algorithm $\mathsf{A}$ solves a problem \textsc{P}$:\mathcal{X}\to2^\mathcal{Y}$ with bounded error $\delta(\abs{x})$ if for all $x\in\mathcal{X}$, $\Pr[\mathsf{A}(x)\in\mbox{\textsc{P}}(x)]\geq 1-\delta(\abs{x})$, where $\abs{x}$ is the size of the input. The quantum query complexity of \textsc{P} is the minimum query complexity of any quantum algorithm that solves \textsc{P} with bounded error $\delta(\abs{x})\leq 1/3$.

We will use the phrase \emph{with high probability} to mean probability at least $1-\frac{1}{\mathrm{poly}}$ for some super-linear polynomial. We ensure that all our subroutines succeed with high probability, to achieve a bounded-error algorithm at the end. To achieve such high probability, we will necessarily incur $\mathrm{polylog}$ factors. We will use the notation $\tO$ to indicate that we are suppressing $\mathrm{polylog}$ factors. More precisely, $f(n) \in \tO(g(n))$ means $f(n) \in \O(g(n)\log^k n)$ for some constant $k$.

\paragraph{Boolean Matrices.}
We let $\BSR$ denote the \emph{Boolean semiring}, which is the set $\{0,1\}$ under the operations $\vee,\wedge$. The problem we will be considering is formally defined as the following:

\begin{quote}
\textsc{Boolean Matrix Multiplication}\\
\textit{Oracle Input:} Two Boolean matrices $A,B\in\BSR^{n\times n}$. \\
\textit{Output:} $C\in\BSR^{n\times n}$ such that $C=AB$.
\end{quote}

In $\BSR^{n\times n}$, we say that $C=AB$ if for all $i,j\in [n]$, $C[i,j]=\bigvee_{k=1}^n A[i,k]\wedge B[k,j]$. We will use the notation $A+B$ to denote the entry-wise $\vee$ of two Boolean matrices.

\subsection{Quantum Search Algorithms \label{sec:search}}

In this section we examine some well-known variations of the search problem that we require. The reader familiar with quantum search algorithms may skip to Section \ref{sec:graph-col}.

Any search problem can be recast as searching for a marked element among a given collection, $U$.
In order to formalize this, let $f:U\to \{0,1\}$ be a function whose purpose is to identify marked elements.
An element is marked if and only if $f(x)=1$. Define $t_f=|f^{-1}(1)|$.
In Grover's search algorithm, the algorithm can directly access $f$, and the overall complexity can be stated as the number of queries to $f$.
In the following $t\geq 1$ is an integer parameter.
\begin{theorem}[\cite{gro96}]\label{thm:grover}
There is a quantum algorithm, $\mathsf{GroverSearch}(t)$, with query complexity $\widetilde{\mathrm{O}}(\sqrt{|U|/t})$ to $f$, 
such that, if $t/2\leq t_f\leq t$, then $\mathsf{GroverSearch}(t)$ finds a marked element with probability at least $1-1/\mathrm{poly}(|U|)$.\\
Moreover, if $t_f=0$, then $\mathsf{GroverSearch}(t)$ declares with probability $1$ that there is no marked element.
\end{theorem}

There are several ways to generalize the above statement when no approximation of $t$ is known.
Most of the generalizations in the literature are stated in terms of expected query complexity, such as in~\cite{BBHT98}.
Nonetheless, one can derive from~\cite[Lemma~2]{BBHT98} an algorithm
in terms of worst case complexity, when only a lower bound $t$ on $t_f$ is known.
The algorithm consists of $T$ iterations of one step of the original Grover algorithm
where~$T$ is chosen uniformly at random from $[0,\sqrt{|U|/t}\,]$. This procedure is iterated $\mathrm{O}(\log |U|)$ times in order to get bounded error $1/\mathrm{poly}(|U|)$.
\begin{corollary}\label{cor:search}
There is a quantum algorithm $\mathsf{Search}(t)$ with query complexity $\widetilde{\mathrm{O}}(\sqrt{|U|/t})$ to $f$,
such that, if $t_f\geq t$, then $\mathsf{Search}(t)$ finds a marked element with probability at least  $1-1/\mathrm{poly}(|U|)$.\\
Moreover, if $t_f=0$, then $\mathsf{Search}(t)$ declares with probability $1$ that there is no marked element.
\end{corollary}
One consequence of Corollary~\ref{cor:search} is that we can always apply $\mathsf{Search}(t)$ with $t=1$,
when no lower bound on $t_f$ is given. In that case, 
we simply refer to the resulting algorithm as $\mathsf{Search}$. 
Its query complexity to $f$ is then $\widetilde{\mathrm{O}}(\sqrt{|U|})$.

Another simple generalization is for finding all marked elements. 
This generalization is stated in the literature in various ways for expected and worst case complexity.
For the sake of clarity we explicitly describe one version of this procedure
using $\mathsf{GroverSearch}$ as a subroutine.
This version is robust in the sense that it works even when the number of marked elements may decrease arbitrarily. This may occur, for example, when the finding of one marked element may cause several others to become unmarked.
This situation will naturally occur in the context of Boolean matrix multiplication.
Then the complexity  will only depend on the number of elements that are actually in the output, as opposed to the number of elements that were marked at the beginning of the algorithm.

\hlight{$\mathsf{SearchAll}\phantom{()}$
\hrule
\vspace{2pt}
\begin{enumerate}
\item Let $t=|U|$, and $V=U$
\item While $t\geq 1$
\begin{enumerate}
\item Apply $\mathsf{GroverSearch}(t)$ to $V$
\item If a marked element $x$ is found:
Output $x$; Set $V\gets V-\{x\}$ and $t\gets t-1$\\
Else: $t\gets t/2$
\end{enumerate}
\item If no marked element has been found, declare `no marked element'
\end{enumerate}
}
\begin{corollary}\label{cor:search-all1}
$\mathsf{SearchAll}$ has query complexity $\widetilde{\mathrm{O}}(\sqrt{|U| (t_f+1)})$ to $f$,
and finds all marked elements with probability at least $1-1/\mathrm{poly}(|U|)$.\\
Moreover, if $t_f=0$, then $\mathsf{SearchAll}$ declares with probability $1$ that there is no marked element.
\end{corollary}

We end this section with an improvement of $\mathsf{GroverSearch}$ when we are looking for an optimal solution for some notion of maximization.
\begin{theorem}[\cite{dh96,dhhm06}]\label{th:find-max}
Given a function $g:U\rightarrow \mathbb{R}$, there is a quantum algorithm, $\mathsf{FindMax}(g)$, with query complexity $\tO(\sqrt{\abs{U}})$ to $f$, such that $\mathsf{FindMax}(g)$ returns $x\in f^{-1}(1)$ such that $g(x)=\max_{x'\in f^{-1}(1)}g(x')$ with probability at least $1-1/\mathrm{poly}(\abs{U})$.
Moreover, if $t_f=0$, then $\mathsf{FindMax}(g)$ declares with probability $1$ that there is no marked element.
\end{theorem}

\section{Graph Collision}\label{sec:graph-col}

In this section we describe the graph collision problem, and its relation to Boolean matrix multiplication. We then describe a method for solving the special case of graph collision in which we are interested.  

\subsection{Problem Description}

Graph collision is the following problem. Let $G=(\set{A},\set{B},E)$ be a balanced bipartite graph on $2n$ vertices. We will suppose $\set{A}=[n]$ and $\set{B}=[n]$, though we note that in the bipartite graph, the vertex labelled by $i$ in $\set{A}$ is distinct from the vertex labelled by $i$ in $\set{B}$. 

\begin{quote}
\textsc{Graph Collision}$(G)$\\
\textit{Oracle Input:} A pair of Boolean functions $\a:\set{A}\rightarrow\{0,1\}$ and $\b:\set{B}\rightarrow\{0,1\}$. \\
\textit{Output:} $(i,j)\in \set{A}\times \set{B}$ such that $\a(i)=\b(j)=1$ and $(i,j)\in E$, if such a pair exists, otherwise reject.
\end{quote}

The graph collision problem was introduced by Magniez, Santha and Szegedy as a subproblem in triangle finding \cite{MSS07}.  The subroutine used to solve an instance of graph collision is based on Ambainis' quantum walk algorithm for element distinctness \cite{amb04}, and has query complexity $\O(n^{2/3})$. The same subroutine is used in the current best triangle finding algorithm of Belovs~\cite{bel11}. However, the best known lower bound for this problem is $\Omega(\sqrt{n})$. It is an important open problem to close this gap. 

To obtain our upper bound, we do not use the quantum walk algorithm for graph collision, but rather, a new algorithm that takes advantage of two special features of our problem. The first is that we always know an upper bound, $\ell$, on the number of non-edges. When $\ell\leq n$, we can find a graph collision in $\O(\sqrt{n})$ queries. The second salient feature of our problem is that we need to find all graph collisions.

\subsection{Relation to Boolean Matrix Multiplication}\label{sec:rel-bmm}

Recall that the Boolean matrix product of $A$ and $B$, can be viewed as the sum (entry-wise $\vee$) of $n$ outer products: $C = \sum_{k=1}^n A[\cdot,k]B[k,\cdot]$, where $A[\cdot,k]$ denotes the \th{k} column of $A$ and $B[k,\cdot]$ denotes the \th{k} row of $B$. 

For a fixed $k$, if there exists some $i\in [n]$ and some $j\in[n]$ such that $A[i,k]=1$ and $B[k,j]=1$, then we know that $C[i,j]=1$, and we say that $k$ is a \emph{witness} for $(i,j)$. We are interested in finding all such pairs $(i,j)$. For each index $k$, we could search for all pairs $(i,j)$ with $A[i,k]=B[k,j]=1$; however, this could be very inefficient, since a pair $(i,j)$ may have up to $n$ witnesses. Instead, we will keep a matrix $\widetilde{C}$ such that $\widetilde{C}[i,j]=1$ if we have already found a one at position $(i,j)$. Thus, we want to find a pair $(i,j)$ such that $A[i,k]=B[k,j]=1$ and $\widetilde{C}[i,j]=0$. That is, we want to find a graph collision in the graph with bi-adjacency matrix $\overline{\widetilde{C}}$, the entry-wise complement of $\widetilde{C}$, and $\a=A[\cdot,k]$, $\b=B[k,\cdot]$.

This gives the following natural algorithm for Boolean matrix multiplication, 
whose details and full analysis can be found in Section~\ref{sec:algo}:
\begin{quote}
First, let $\widetilde{C} = 0$.\\
Search for an index $k$ such that the graph collision problem on $k$ with $\overline{\widetilde{C}}$ as the underlying graph has a collision.\\
If no such $k$ is found then we are done, and  $\widetilde{C}$ is the product of $A$ and $B$.\\
Otherwise,  find all the graph collisions on the graph defined by $\overline{\widetilde{C}}$ with oracles $A[\cdot,k]$ and $B[k,\cdot]$ and record them in $\widetilde{C}$.\\
Eliminate this $k$ from future searches and search for another index $k$ again.
\end{quote}

\subsection{Algorithm for Graph Collision}

When $G$ is a complete bipartite graph, then the relation between $\set{A}$ and $\set{B}$ defined by $G$ is trivial. In that case, there is a very simple algorithm to find a graph collision: Search for some $i\in [n]$ such that $\a(i)=1$. Then search for some $j\in [n]$ such that $\b(j)=1$. Then $(i,j)$ is a graph collision pair. The query complexity of this is $\O(\sqrt{n}+\sqrt{n})$. 
However, when $G$ is not a complete bipartite graph, there is a nontrivial relation between $\set{A}$ and $\set{B}$. The best known algorithm solves this problem using a quantum walk. 

In our case, we can take advantage of the fact that the graph we are working with always has at most $\ell$ non-edges --- it is never more than distance $\ell$ from the complete bipartite graph, which we know is easy to deal with. We are therefore interested in the query complexity of finding a graph collision in some graph with $m$ non-edges, which we denote $\cost{GC}(n,m)$. In our case, $\ell$ will always be an upper bound on $m$.

For larger values of $\ell$, we will also make use of the fact that for some $k$, we will have multiple graph collisions to find. We let $\cost{GC}_\text{all}(n,m,\lambda)$ denote the query complexity of finding all graph collisions in a graph with $m$ non-edges, where $\lambda$ is the number of graph collisions. It is not necessary to know $\lambda$ a priori.

Again we note that if $G$ is a complete bipartite graph, then we can accomplish the task of finding all graph collisions using $\mathsf{SearchAll}$ to search for all marked elements on each of $\a$ and $\b$, and output $\a^{-1}(1)\times \b^{-1}(1)$. Letting $t_A=\abs{\a^{-1}(1)}$ and $t_B=\abs{\b^{-1}(1)}$, so the total number of graph collision pairs is $\lambda=t_At_B$, the query complexity of this method is $\O(\sqrt{nt_A}+\sqrt{nt_B})\in \O(\sqrt{n\lambda})$. So if $G$ is close to being a complete bipartite graph, we would like to argue that we can do nearly as well. This motivates the following algorithm.

\hlight{$\mathsf{AllGC}_G(\a,\b)$
\hrule
\vspace{2pt}
Let $\m$ denote the number of non-edges in $G$. Let $d_i$ be the degree of the \th{i} vertex in $\mathcal{A}$, and let $c_i \defeq n - d_i$. Let the vertices in $\mathcal{A}$ be arranged in decreasing order of degree, so that $d_1 \geq d_2 \geq \ldots \geq d_n$. We will say a vertex $i$ in $\mathcal{A}$ (resp. $\set{B}$) is \emph{marked} if $\a(i)=1$ (resp. $\b(i)=1$).
\begin{enumerate}
  \item Find the highest degree marked vertex in $\mathcal{A}$ using $\mathsf{FindMax}$. Let $r$ denote the index of this vertex. \hfill $\tO(\sqrt{n})$
  
\item Case 1: If $c_r \leq \sqrt{\m}$
	\begin{enumerate}
    \item\label{it:2a} Find all marked neighbors of $r$ by $\mathsf{SearchAll}$. Output any graph collisions found.  \hfill $\tO(\sqrt{n\lambda})$
    \item Delete all unmarked neighbors of $r$. Read the values of all non-neighbors of $r$. \hfill $\O(\sqrt{\m})$
	 \item\label{it:2c} Let $\mathcal{A}'$ denote the subset of $\mathcal{A}$ consisting of all $i\in\mathcal{A}$ with a marked neighbour in $\mathcal{B}$. Find all marked vertices in $\mathcal{A}'$ by $\mathsf{SearchAll}$. \hfill $\tO(\sqrt{n\lambda})$
  \end{enumerate}
\item Case 2: If $c_r \geq \sqrt{\m}$
	\begin{enumerate}
    \item Delete the first $r-1$ vertices in $\mathcal{A}$ since they are unmarked. 
    \item Read the values of all remaining vertices in $\mathcal{A}$. \hfill $\O(\sqrt{\m})$
	 \item\label{it:3c} Let $\mathcal{B}'$ denote the subset of $\mathcal{B}$ consisting of all $j\in\mathcal{B}$ with a marked neighbour in $\mathcal{A}$. Find all marked vertices in $\mathcal{B}'$ by $\mathsf{SearchAll}$. \hfill $\tO(\sqrt{n\lambda})$
  \end{enumerate}
\end{enumerate}
}

\begin{theorem}
For all $\lambda\geq 1$, $\cost{GC}_\text{\emph{all}}(n,\m,\lambda) \in \tO(\sqrt{n\lambda} + \sqrt{\m})$ and $\cost{GC}(n,\m) \in \tO(\sqrt{n} + \sqrt{\m})$. 
\end{theorem}
\begin{proof}
We will analyze the complexity of $\mathsf{AllGC}_G(\a,\b)$ step by step.

Step 1 has query complexity $\tO(\sqrt{n})$ by Theorem \ref{th:find-max}. Steps 2a, 2c and 3c have query complexity $\tO(\sqrt{n\lambda})$ by Corollary \ref{cor:search-all1}. In Case 1, $r$ has $c_r\leq\sqrt{\m}$ non-neighbours, so we can certainly query them all in step 2b with $\O(\sqrt{c_r})\in\O(\sqrt{\m})$ queries.

Consider Case 2, when $c_r\geq\sqrt{\m}$. We can ignore the first $r-1$ vertices, since they are unmarked. Since the remaining $n-r+1$ vertices all have $c_i \geq c_r \geq \sqrt{\m}$, and the total number of non-edges is $\m$, we have $(n-r+1)\times \sqrt{\m} \leq \m \Rightarrow (n-r+1) \leq \sqrt{\m}$. Thus, there are at most $\sqrt{\m}$ remaining vertices and querying them all costs at most $O(\sqrt{m})$ queries.

The query complexity of this algorithm is therefore $\tO(\sqrt{n\lambda} + \sqrt{m})$, and it outputs all graph collisions. To check if there is at least one graph collision, instead of finding them all, we can replace finding all marked vertices using $\mathsf{SearchAll}$ in steps 2a, 2c and 3c, with a procedure to check if there is any marked vertex, $\mathsf{Search}$, and this only requires $\tO(\sqrt{n})$ queries by Corollary \ref{cor:search}, rather than $\tO(\sqrt{n\lambda})$.
\end{proof}

\section{Boolean Matrix Multiplication}\label{sec:bmm}

In this section we show how the graph collision algorithm from the previous section can be used to obtain an efficient algorithm for Boolean matrix multiplication and then prove a lower bound.

\subsection{Algorithm \label{sec:algo}}

What follows is a more precise statement of the high level procedure described in Section \ref{sec:rel-bmm}.
\hlight{$\mathsf{BMM}(A,B)$
\hrule
\vspace{2pt}
\begin{enumerate}
  \item Let $\widetilde{C} = 0$, $t=n$, and $V=[n]$
  \item While $t\geq 1$
  \begin{enumerate}
  \item $\mathsf{GroverSearch}(t)$ for an index $k\in V$ such that the graph collision problem on $k$ with $\overline{\widetilde{C}}$ as the underlying graph has a collision. 
  \item If such a $k$ is found
\begin{itemize}
  \item[] Compute $\mathsf{AllGC}$ on the graph defined by $\overline{\widetilde{C}}$ with oracles $A[\cdot,k]$ and $B[k,\cdot]$ and record all output graph collisions in $\widetilde{C}$. 
  \item[] Set $V\gets V-\{k\}$ and $t\gets t-1$.
\end{itemize}
\item Else: $t\gets t/2$
\end{enumerate}
\item Output $\widetilde{C}$.
\end{enumerate}
}

\begin{theorem}
The query complexity of \textsc{Boolean Matrix Multiplication} is $\tO(n\sqrt{\ell})$.
\end{theorem}
\begin{proof}
We will analyze the complexity of the algorithm $\mathsf{BMM}(A,B)$. 
We begin by analyzing the cost of all the iterations in which we don't find a marked $k$. We have by Theorem \ref{thm:grover} that $\mathsf{GroverSearch}(t)$ costs $\tO(\sqrt{n/t})$ queries to a procedure that checks if there is a collision in the graph defined by $\overline{\widetilde{C}}$ with respect to $A[\cdot,k]$ and $B[k,\cdot]$, each of which costs $\cost{GC}(n,m_i)$, where $m_i\leq\ell$ is the number of $1$s in $\widetilde{C}$ at the beginning of the \th{i} iteration. The cost of these steps is at most the following:
$$\tO\(\sum_{i=0}^{\log n}\sqrt{\frac{n}{2^i}}\cost{GC}(n,m_i)\)\in\tO\(\sum_{i=0}^{\log n}\sqrt{\frac{n}{2^i}}(\sqrt{n}+\sqrt{m_i})\)$$
$$\in\tO\((n+\sqrt{n\ell})\sum_{i=0}^{\log n}\(\frac{1}{\sqrt{2}}\)^i\)\in\tO\(n+\sqrt{n\ell}\)$$

We now analyze the cost of all the iterations in which we do find a marked witness $k$. Let $T$ be the number of witnesses found by $\mathsf{BMM}$. Of course, $T$ is a random variable that depends on which witnesses $k$ are found, and in which order. We always have $T\leq\min\{n,\ell\}$. 

Let $i_1,\dots,i_T$ be the indices of rounds where we find a witness. Let $t_j$ be the value of $t$ in round $j$.
Since there must be at least 1 marked element in the last round in which we find a marked element, we have $t_{i_T} \geq 1$. Since we find and eliminate at least 1 marked element in each round, we also have $t_{i_{(T-j-1)}} \geq t_{i_{(T - j)}} +1$, which yields $t_{i_{(T-j)}} \geq j + 1 \Rightarrow t_{i_j} \geq T-j+1$.

Let $\lambda_j$ be the number of graph collisions found on the \th{j} successful iteration, that is, the number of pairs witnessed by the \th{j} witness, $k_j$, that have not been recorded in $\widetilde{C}$ \emph{at the time we find $k_j$}. Then $\lambda_j$ is also a random variable depending on which other witnesses $k$ have been found already, but we always have $\sum_{j=1}^{T}\lambda_j=\ell$. 

Then we can upper bound the cost of all the iterations in which we do find a witness by the following:
\begin{eqnarray}
&&\tO\left(\sum_{j=1}^{T}\(\sqrt{\frac{n}{t_{i_j}}}\cost{GC}(n,m_{i_j})+\cost{GC}_{\text{all}}(n,m_{i_j},\lambda_j)\)\right)\\
&\in&\tO\left(\sum_{j=1}^{T}\(\sqrt{\frac{n}{T-j+1}}\cost{GC}(n,m_{i_j})+\cost{GC}_{\text{all}}(n,m_{i_j},\lambda_j)\)\right)\\
&\in& \tO\(\sqrt{nT}\cost{GC}(n,m_{i_j})+\sum_{j=1}^{T}\cost{GC}_\text{all}(n,m_{i_j},\lambda_j)\)\\
&\in& \tO\(\sqrt{nT}\(\sqrt{\ell}+\sqrt{n}\)+\sum_{j=1}^{T}\(\sqrt{n\lambda_j}+\sqrt{\ell}\)\)\\
&\in& \tO\(\sqrt{nT\ell}+n\sqrt{T}+\sqrt{n{\ell}{T}}+T\sqrt{\ell}\)\\
&\in& \tO\(\sqrt{n\min\{n,\ell\}\ell}+n\sqrt{\min\{n,\ell\}}+\sqrt{\min\{n,\ell\}n\ell}+\min\{n,\ell\}\sqrt{\ell}\)\\
&\in& \tO\(n\sqrt{\ell}\)
\end{eqnarray}
In (4), we use the fact that $m_{i_j}\leq \ell$, and in (5), we use $\sum_{j=1}^T \sqrt{\lambda_j} \leq \sqrt{T}\sqrt{\sum_j \lambda_j} = \sqrt{\ell T}$, which follows from the Cauchy--Schwarz inequality.
\end{proof}

\subsection{Lower Bound}

\begin{theorem}\label{th:lb}
The query complexity of \textsc{Boolean Matrix Multiplication} is $\Omega(n\sqrt{\min\{\ell,n^2-\ell\}})$.
\end{theorem}
\begin{proof}
We will reduce the problem of $\ell$-\textsc{Threshold}, in which we must determine whether an input oracle $f$ has $\geq \ell$ or $<\ell$ marked elements, to \textsc{Boolean Matrix Multiplication}. 

Consider an instance of $\ell$-\textsc{Threshold} of size $n^2$, $f:[n^2]\rightarrow\{0,1\}$. We can construct an instance of \textsc{Boolean Matrix Multiplication} as follows. Set $A$ to the identity, and let $B$ encode $f$. Finding $AB$ then gives the solution to the $\ell$-\textsc{Threshold} instance. 
 By \cite{BBCMdW01}, $\ell$-\textsc{Threshold} (with inputs of size $n^2$) requires $\Omega(\sqrt{n^2\min\{\ell,n^2-\ell\}})$ queries to solve with bounded error. 
\end{proof}

This lower bound implies that our algorithm is tight for any $\ell\leq \eps n^{2}$ for any constant $\eps<1$. However, it is not tight for $\ell=n^2-o(n)$. We can search for pairs $(i,j)$ such that there is no $k\in [n]$ that witnesses $(i,j)$ in cost $n^{3/2}$. If there are $m$ $0$s, we can find them all in $\tO(n^{3/2}\sqrt{m})$, which is $o(n\sqrt{\ell})$ when $m\in o(n)$. It remains open to close the gap between $\tO(n^{3/2}\sqrt{m})$ and $\Omega(n\sqrt{m})$ when $m\in o(n^2)$. 

\section{Acknowledgements}

This work was partially supported by NSERC, MITACS, QuantumWorks, the French ANR Defis project ANR-08-EMER-012 (QRAC), and the European Commission IST STREP project 25596 (QCS).

\bibliographystyle{alpha}
\bibliography{references}

\newcommand{\etalchar}[1]{$^{#1}$}
\begin{thebibliography}{DHHM06}

\bibitem[Amb04]{amb04}
A.~Ambainis.
\newblock Quantum walk algorithm for element distinctness.
\newblock In {\em Proceedings of the 45th IEEE Symposium on Foundations of
  Computer Science}, pages 22--31, 2004.

\bibitem[BBBV97]{BBBV97}
C.~H. Bennett, E.~Bernstein, G.~Brassard, and U.~Vazirani.
\newblock Strengths and weaknesses of quantum computing.
\newblock {\em SIAM Journal on Computing (special issue on quantum computing)},
  26:1510--1523, 1997.
\newblock {\tt arXiv:quant-ph/9701001v1}.

\bibitem[BBC{\etalchar{+}}01]{BBCMdW01}
R.~Beals, H.~Buhrman, R.~Cleve, M.~Mosca, and R.~de~Wolf.
\newblock Quantum lower bounds by polynomials.
\newblock {\em Journal of the ACM}, 48, 2001.

\bibitem[BBHT98]{BBHT98}
M.~Boyer, G.~Brassard, P.~H{\o}yer, and A.~Tapp.
\newblock Tight bounds on quantum searching.
\newblock {\em Fortschritte der Physik}, 46(4-5):493--505, 1998.

\bibitem[Bel11]{bel11}
A.~Belovs.
\newblock Span programs for functions with constant-sized 1-certificates.
\newblock Technical Report arXiv:1105.4024, arXiv, 2011.

\bibitem[B{\v{S}}06]{BS06}
H.~Buhrman and R.~{\v{S}}palek.
\newblock Quantum verification of matrix products.
\newblock In {\em Proceedings of the 17th ACM-SIAM Symposium On Discrete
  Algorithms}, pages 880--889, 2006.

\bibitem[CKK11]{CKK11}
A.~Childs, S.~Kimmel, and R.~Kothari.
\newblock The quantum query complexity of read-many formulas.
\newblock Technical Report arXiv:1112.0548v1, arXiv, 2011.

\bibitem[DH96]{dh96}
C.~D{\"u}rr and P.~H{\o}yer.
\newblock A quantum algorithm for finding the minimum.
\newblock Technical Report arXiv:quant-ph/9607014v2, arXiv, 1996.

\bibitem[DHHM06]{dhhm06}
C.~D{\"u}rr, M.~Heiligman, P.~H{\o}yer, and M.~Mhalla.
\newblock Quantum query complexity of some graph problems.
\newblock {\em SIAM Journal on Computing}, 35(6):1310--1328, 2006.

\bibitem[Gal12]{leg12}
F.~Le Gall.
\newblock Improved output-sensitive quantum algorithms for {B}oolean matrix
  multiplication.
\newblock In {\em Proceedings of the 23rd ACM-SIAM Symposium On Discrete
  Algorithms}, pages 1464--1476, 2012.

\bibitem[Gro96]{gro96}
L.~K. Grover.
\newblock A fast quantum mechanical algorithm for database search.
\newblock In {\em Proceedings of the 28th ACM Symposium on Theory of
  Computing}, pages 212--219, 1996.

\bibitem[Lin09]{lin09}
A.~Lingas.
\newblock A fast output-sensitive algorithm for {B}oolean matrix
  multiplication.
\newblock In {\em Proceedings of the 17th European Symposium on Algorithms},
  pages 408--419, 2009.

\bibitem[MSS07]{MSS07}
F.~Magniez, M.~Santha, and M.~Szegedy.
\newblock Quantum algorithms for the triangle problem.
\newblock {\em SIAM Journal on Computing}, 37(2):413--424, 2007.

\bibitem[Sho97]{sho97}
P.~W. Shor.
\newblock Polynomial-time algorithms for prime factorization and discrete
  logarithms on a quantum computer.
\newblock {\em SIAM Journal on Computing}, 26:1484--1509, October 1997.

\bibitem[VW10]{VW10}
V.~{Vassilevska Williams} and R.~Williams.
\newblock Sub-cubic equivalences between path, matrix and triangle problems.
\newblock In {\em Proceedings of the 51st IEEE Symposium on Foundations of
  Computer Science}, pages 645--654, 2010.

\end{thebibliography}

\end{document}